\documentclass[11pt]{article}%
\usepackage{amsmath}
\usepackage{amsfonts}
\usepackage{amssymb}
\usepackage{graphicx}%
\newtheorem{theorem}{Theorem}

\newtheorem{corollary}[theorem]{Corollary}

\newtheorem{definition}[theorem]{Definition}
\newtheorem{example}[theorem]{Example}

\newenvironment{proof}[1][Proof]{\noindent\textbf{#1.} }{\ \rule{0.5em}{0.5em}}
\begin{document}
\centerline{{\large \textbf{Semiconjugate Factorization and Reduction of Order} }}
\vspace{1ex}
\centerline{{\large \textbf{in Difference Equations} }}

\vspace{4ex}

\centerline{H. SEDAGHAT*} \footnotetext{\noindent*Department of Mathematics,
Virginia Commonwealth University, Richmond, Virginia 23284-2014
\par
USA, Email: hsedagha@vcu.edu}

\vspace{3ex}

\begin{quote}
\noindent{\small \textbf{Abstract.} We discuss a general method by which a
higher order difference equation on a group is transformed into an equivalent
triangular system of two difference equations of lower orders. This breakdown
into lower order equations is based on the existence of a semiconjugate
relation between the unfolding map of the difference equation and a lower
dimensional mapping that unfolds a lower order difference equation.
Substantial classes of difference equations are shown to possess this property
and for these types of equations reductions of order are obtained. In some
cases a complete semiconjugate factorization into a triangular system of first
order equations is possible.}
\end{quote}

\section{Introduction}

There are a number of known methods by which a difference equation can be
transformed into equations with lower orders. These methods are almost
universally adapted from the theory of differential equations and can be
divided into two major categories. One category uses symmetries in the
solutions of a difference equation to obtain coordinate transformations that
lead to a reduction in order; see, e.g., \cite{bsq}, \cite{hyd}, \cite{ltw},
\cite{mae}. The second category consists of methods (e.g., operator methods)
that rely on the algebraic and analytical properties of the difference
equation itself; see, e.g., \cite{el}, \cite{jor}. For historical reasons the
first of these categories has been developed considerably more than the second.

In this article we introduce a general new method that falls into the second
category. This method uses the existence of a semiconjugate relation to break
down a given difference equation into an equivalent pair of lower order
equations whose orders add up to the order of the original equation. A salient
feature of the equivalent system of two equations is its \textquotedblleft
triangular\textquotedblright\ nature; see \cite{al} for a formal definition as
well as a study of the periodic structure of solutions. Also see \cite{js} for additional information and a list of references. Such a system is
uncoupled in the sense that one of the two equations is independent of the other.

In some cases (e.g., linear equations) we use the semiconjugate relation
repeatedly to transform a higher order difference equation into an equivalent
system of first order equations with triangular uncoupling. In the specific
case of linear equations, the system of first order equations explicitly
reveals all of the eigenvalues of the higher order equation.

Semiconjugate-based reduction in order uses only patterns or symmetries
inherent in the \textit{form} of the difference equation itself. For this
reason, we do not limit ourselves to equations defined on real numbers.
Allowing more general types of objects not only improves the applicability of
the method to areas where discrete modeling naturally occurs (e.g., dynamics
on networks, see for instance, \cite{mr}) but it is actually helpful in
identifying form symmetries even within the context of real numbers where more
than one algebraic operation is defined. After defining the basic idea of
semiconjugate factorization and laying some groundwork we discuss applications
of the method to reduction of order in several classes of difference equations.

\section{Difference equations and semiconjugacy}

Difference equations of order greater than one that are of the following type%
\begin{equation}
x_{n+1}=f_{n}(x_{n},x_{n-1},\ldots,x_{n-k}) \label{dek}%
\end{equation}
determine the forward evolution of a variable $x_{n}$ in discrete time. As is
customary, the time index or the \textit{independent variable} $n$ is
integer-valued, i.e., $n\in\mathbb{Z}$. The number $k$ is a fixed positive
integer and $k+1$ represents the \textit{order} of the difference equation
(\ref{dek}). If the underlying space of variables $x_{n}$ is labled $G$ then
$f_{n}:G^{k+1}\rightarrow G$ is a given function for each $n\geq1$. If
$f_{n}=f$ does not explicitly depend on $n$ then (\ref{dek}) is said to be
\textit{autonomous}. A \textit{solution} of Eq.(\ref{dek}) is any function
$x:\mathbb{Z}\rightarrow G$ that satisfies the difference equation. A
\textit{forward} solution is a sequence $\{x_{n}\}_{n=-k}^{\infty}$ that is
recursively generated by (\ref{dek}) from a set of $k+1$ initial values
$x_{0},x_{-1},\ldots,x_{-k}\in G.$ Forward solutions have traditionally been
of greater interest in discrete models that are based on Eq.(\ref{dek}).

If the functions $f_{n}$ do not have any singularities in $G$ \ for all $n$
then clearly the existence of forward solutions is guaranteed from any given
set of initial values. Otherwise, sets of initial values that upon a finite
number of iterations lead to a singularity of $f_{n}$ for some $n$ must either
be specified or avoided. In this paper we do not discuss the issue of
singularities broadly but deal with them as they occur in various examples.

We assume that the underlying space $G$ of variables $x_{n}$ is endowed with
an algebraic \textit{group} structure; for a study of semiconjugacy in a
topological context see \cite{hsbk} or \cite{scjd}. Often there are multiple group operations
defined on a given set $G$. This feature turns out to be quite useful.
Compatible topologies (making group operation(s) continuous) may exist on $G$
and often occur in applications. Topological or metric concepts may be used
explicitly in defining certain form symmetries and are required in discussing
the asymptotic behaviors of solutions with unknown closed forms in infinite sets.

We \textquotedblleft unfold\textquotedblright\ each $f_{n}$ in the usual way
using the functions $F_{n}:G^{k+1}\rightarrow G^{k+1}$ defined as%
\begin{equation}
F_{n}(u_{0},\ldots,u_{k})=[f_{n}(u_{0},\ldots,u_{k}),u_{0},\ldots
,u_{k-1}],\quad u_{j}\in G\text{ for }j=0,1,\ldots,k. \label{unfold}%
\end{equation}

The \textit{unfoldings} (or associated vector maps) $F_{n}$ determine the
equation
\[
(y_{0,n+1},y_{1,n+1},\ldots,y_{k,n+1})=F_{n}(y_{0,n},y_{1,n},\ldots,y_{k,n})
\]
in $G^{k+1}$ or equivalently, the system of equations%
\[
y_{j,n+1}=\left\{
\begin{array}
[c]{l}%
f_{n}(y_{0,n},y_{1,n},\ldots,y_{k,n}),\text{ if }j=0\\
y_{j-1,n},\text{ for }j=1,\ldots,k
\end{array}
\right.  .
\]

In this context we refer to $(y_{0,n},y_{1,n},\ldots,y_{k,n})$ as the
\textit{state of the system} at time $n$ and to $G^{k+1}$ as the \textit{state
space} of the system, or of (\ref{dek}).

\begin{definition}
\label{scmaps} Let $S$ and $M$ be arbitrary nonempty sets and let $F,\Phi$ be
self maps of $S$ and $M$, respectively. If there is a surjective (onto)
mapping $H:S\rightarrow M$ such that%
\begin{equation}
H\circ F=\Phi\circ H \label{sconj}%
\end{equation}
then we say that the mapping $F$ is \textbf{semiconjugate} to $\Phi$ and write
$F\trianglerighteq\Phi$. We refer to $\Phi$ as a semiconjugate (SC)
\textbf{factor} of $F$. The function $H$ may be called a \textbf{link map}. If
$H$ is a bijection (one to one and onto) then we call $F$ and $\Phi$
\textbf{conjugates} and alternatively write%
\begin{equation}
\Phi=H\circ F\circ H^{-1}. \label{conj}%
\end{equation}

\end{definition}

We use the notation $F\simeq\Phi$ when $F$ and $\Phi$ are conjugates.
Straightforward arguments show that $\simeq$ is an equivalence relation for
self maps and that $\trianglerighteq$ is a transitive relation.

As our first example shows, the fact that the link map $H$ is \textit{not}
injective or one-to-one is essential for reduction of order. Such a
non-conjugate map often exists where conjugacy does not.

\begin{example}
\label{squares}Consider the following self maps of the plane $\mathbb{R}^{2}$:%
\[
F_{1}(u,v)=[u^{2}+v^{2},2uv],\quad F_{2}(u,v)=[u^{2}-v^{2},2uv].
\]
Note that $u^{2}+v^{2}+2uv=(u+v)^{2}$ so let $H_{1}(u,v)=u+v$ to obtain%
\[
H_{1}(F_{1}(u,v))=(u+v)^{2}=[H_{1}(u,v)]^{2}.
\]
It follows that $F_{1}$ is semiconjugate to $\Phi(t)=t^{2}$ on $\mathbb{R}.$
For $F_{2}$ define $H_{2}(u,v)=(u^{2}+v^{2})^{1/2}$ to get%
\[
H_{2}(F_{2}(u,v))=\sqrt{\left(  u^{2}-v^{2}\right)  ^{2}+4u^{2}v^{2}}%
=u^{2}+v^{2}=[H_{2}(u,v)]^{2}.
\]
Therefore, $F_{2}$ is also semiconjugate to $\Phi(t)=t^{2}$ though this time
on $[0,\infty)$ since $H_{2}(u,v)\geq0$ for all $(u,v)\in$ $\mathbb{R}^{2}$.
\end{example}

In the next example we determine all possible semiconjugates for a difference
equation on a finite field.

\begin{example}
\label{z2b}Consider the following autonomous difference equation of order 2 on
$\mathbb{Z}_{2}=\{0,1\}$%
\[
x_{n+1}=(x_{n}+1)x_{n-1},\quad x_{0},x_{-1}\in\{0,1\}
\]
Note that $\mathbb{Z}_{2}$ is a field under addition modulo 2 and ordinary
multiplication of integers. The unfolding of the defining map $f(u,v)=(u+1)v$
is given by $F(u,v)=[(u+1)v,u]$. In order to determine the semiconjugates of
$F$ on $\mathbb{Z}_{2}^{2}$ first we list the four possible self maps of
$\mathbb{Z}_{2}$, i.e.,%
\begin{align*}
\phi_{0}  &  \equiv0,\quad\phi_{1}\equiv1\quad\text{(constant functions)}\\
\iota(t)  &  =t\quad\text{(the identity function)}\\
\phi(t)  &  =t+1.
\end{align*}
Next, there are 16 link maps $H:\mathbb{Z}_{2}^{2}\rightarrow\mathbb{Z}_{2}$
of which the 14 non-constant ones are surjective. We list these 14 maps
succinctly as follows:%
\[%
\begin{array}
[c]{cccc}%
H_{1}=(0,0,0,1) & H_{2}=(0,0,1,0) & H_{3}=(0,0,1,1) & H_{4}=(0,1,0,0)\\
H_{5}=(0,1,0,1) & H_{6}=(0,1,1,0) & H_{7}=(0,1,1,1) & H_{8}=(1,0,0,0)\\
H_{9}=(1,0,0,1) & H_{10}=(1,0,1,0) & H_{11}=(1,0,1,1) & H_{12}=(1,1,0,0)\\
H_{13}=(1,1,0,1) & H_{14}=(1,1,1,0) &  &
\end{array}
\]
In the above list, for instance $H_{14}$ is defined by the rule%
\[
H_{14}:(0,0)\rightarrow1,\ (0,1)\rightarrow1,\ (1,0)\rightarrow
1,\ (1,1)\rightarrow0
\]
and can be written in algebraic form as $H_{14}(u,v)=uv+1.$ The other link
maps are defined in the same way. Now writing $F$ as%
\[
F:(0,0)\rightarrow(0,0),\ (0,1)\rightarrow(1,0),\ (1,0)\rightarrow
(0,1),\ (1,1)\rightarrow(0,1)
\]
and calculating the compositions $H_{j}\circ F$ for $j=1,\ldots,14$ the
following results are obtained:%
\begin{align*}
H_{1}\circ F  &  =(0,0,0,0)=\phi_{0}\circ H_{1},\quad H_{14}\circ
F=(1,1,1,1)=\phi_{1}\circ H_{14}\\
H_{7}\circ F  &  =(0,1,1,1)=H_{7}\circ\iota,\quad H_{8}\circ F=(1,0,0,0)=H_{8}%
\circ\iota.
\end{align*}
Therefore, $F$ is semiconjugate to each of $\phi_{0},\phi_{1}$ and $\iota$ but
not to $\phi$ (but also see Example \ref{z2c} below).
\end{example}

It is worth noting that the maps $H_{7}$ and $H_{8}$ are \textit{invariants}
(see, e.g., \cite{BT}, \cite{GKL}, \cite{sah}) because the factor map is the
identity function ($\iota$ here). It is easy to check that in algebraic form
the two invariants can be written as
\[
H_{7}(u,v)=uv+u+v,\quad H_{8}(u,v)=uv+u+v+1=(u+1)(v+1).
\]

\section{Semiconjugate factorization}

In this section we define the concepts needed to present our results in
subsequent sections.

\subsection{Form symmetries, factors and cofactors}

To make the transition from semiconjugate maps to reduction of order in
difference equations, let $k$ be a non-negative integer and $G$ a nonempty
set. Let $\{F_{n}\}$ be a family of functions $F_{n}:G^{k+1}\rightarrow
G^{k+1}$ where $F_{n}=[f_{1,n},\ldots,f_{k+1,n}]$ \ and $f_{j,n}:G\rightarrow
G$ are the component functions of $F_{n}$ for all $j$ and all $n.$

Let $m$ be an integer, $1\leq m\leq k+1$ and assume that each $F_{n}$ is
semiconjugate to a map $\Phi_{n}:G^{m}\rightarrow G^{m}$. Let $H:G^{k+1}%
\rightarrow G^{m}$ be the link map such that for every $n,$%
\begin{equation}
H\circ F_{n}=\Phi_{n}\circ H. \label{sc0}%
\end{equation}

Suppose that%

\begin{align*}
H(u_{0},\ldots,u_{k})  &  =[h_{1}(u_{0},\ldots,u_{k}),\ldots,h_{m}%
(u_{0},\ldots,u_{k})]\\
\Phi_{n}(t_{1},\ldots,t_{m})  &  =[\phi_{1,n}(t_{1},\ldots,t_{m}),\ldots
,\phi_{m,n}(t_{1},\ldots,t_{m})]
\end{align*}
where $h_{j}:G^{k+1}\rightarrow G$ and $\phi_{j,n}:G\rightarrow G$ are the
corresponding component functions for $j=1,2,\ldots,m$. Then identity
(\ref{sc0}) is equivalent to the system of equations%
\begin{gather}
h_{j}(f_{1,n}(u_{0},\ldots,u_{k}),\ldots,f_{k+1,n}(u_{0},\ldots,u_{k}%
))=\nonumber\\
\phi_{j,n}(h_{1}(u_{0},\ldots,u_{k}),\ldots,h_{m}(u_{0},\ldots,u_{k})),\quad
j=1,2,\ldots,m. \label{scj}%
\end{gather}

If the functions $f_{j,n}$ are given then (\ref{scj}) is a system of
\textit{functional equations} whose solutions $h_{j},\phi_{j,n}$ give the maps
$H$ and $\Phi_{n}.$ Note that if $m<k+1$ then the functions $\Phi_{n}$ on
$G^{m}$ define a system with lower dimension than that defined by the
functions $F_{n}$ on $G^{k+1}.$ We now proceed to find a solution set for
(\ref{scj}).

Let $G$ be a group and denote its operation by $\ast.$ If $F_{n}$ is the
unfolding of the function $f_{n}$ as in (\ref{unfold}) and $F_{n}$ satisfies
(\ref{sc0}) then it need not follow that the maps $\Phi_{n}$ are of scalar
type similar to $F_{n}.$ To ensure that each $\Phi_{n}$ is also of scalar type
we define the first component function of $H$ as%
\begin{equation}
h_{1}(u_{0},\ldots,u_{k})=u_{0}\ast h(u_{1},\ldots,u_{k}) \label{H}%
\end{equation}
where $h:G^{k}\rightarrow G$ is a function to be determined.

Using (\ref{H}) in the first equation in (\ref{scj}) gives%
\begin{gather}
f_{n}(u_{0},\ldots,u_{k})\ast h(u_{0},\ldots,u_{k-1})=\label{scm}\\
g_{n}(u_{0}\ast h(u_{1},\ldots,u_{k}),h_{2}(u_{0},\ldots,u_{k}),\ldots
,h_{m}(u_{0},\ldots,u_{k}))\nonumber
\end{gather}
where the abbreviation $g_{n}$ is used for $\phi_{1,n}:G^{m}\rightarrow G.$

Eq.(\ref{scm}) is a functional equation in which the functions $h,h_{j},g_{n}$
may be determined in terms of the given functions $f_{n}$. Our aim is to
extract a scalar equation of order $m$ such as%
\begin{equation}
t_{n+1}=g_{n}(t_{n},\ldots,t_{n-m+1}) \label{gnt}%
\end{equation}
from (\ref{scm}) in such a way that the maps $\Phi_{n}$ will be of scalar
type. The basic framework for carrying out this process is already in place;
let $\{x_{n}\}$ be a solution of Eq.(\ref{dek}) and define%
\[
t_{n}=x_{n}\ast h(x_{n-1},\ldots,x_{n-k}).
\]

Then the left hand side of\ (\ref{scm}) is%
\begin{equation}
x_{n+1}\ast h(x_{n},\ldots,x_{n-k+1})=t_{n+1}. \label{nplus1}%
\end{equation}
which gives the initial part of the difference equation (\ref{gnt}). In order
that the right hand side of (\ref{scm}) coincide with that in (\ref{gnt}) it
is necessary to define the component functions in (\ref{scm}) as%

\begin{align}
h_{j}(x_{n},\ldots,x_{n-k})  &  =t_{n-j+1}=x_{n-j+1}\ast h(x_{n-j}%
,\ldots,x_{n-k-j+1}),\label{hgj}\\
\text{for }j  &  =2,\ldots,m.\nonumber
\end{align}

Since the left hand side of (\ref{hgj}) does not depend on the terms
\[
x_{n-k-1},\ldots,x_{n-k-j+1}%
\]
it follows that the function $h_{j}$ must be constant in its last few
coordinates. Since $h$ does not depend on $j,$ the number of its constant
coordinates is found from the last function $h_{m}.$ Specifically, we have%
\begin{equation}
h_{m}(x_{n},\ldots,x_{n-k})=x_{n-m+1}\ast h(\underset{k-m+1\text{ variables}%
}{\underbrace{x_{n-m},\ldots,x_{n-k}}},\underset{m-1\text{ terms with }h\text{
constant}}{\underbrace{x_{n-k-1},\ldots,x_{n-k-m+1}}}) \label{consis}%
\end{equation}

The preceding condition leads to the necessary restrictions on $h$ and every
$h_{j}$ for a consistent derivation of (\ref{gnt}) from (\ref{scm}).
Therefore, (\ref{consis}) is a consistency condition. Now from (\ref{hgj}) and
(\ref{consis}) we obtain for $(u_{0},\ldots,u_{k})\in G^{k+1}$
\begin{equation}
h_{j}(u_{0},\ldots,u_{k})=u_{j-1}\ast h(u_{j},u_{j+1}\ldots,u_{j+k-m}),\quad
j=1,\ldots,m. \label{hjd}%
\end{equation}

With this definition for the coordinate functions $h_{j}$ we rewrite
(\ref{scm}) for later reference as%

\begin{gather}
f_{n}(u_{0},\ldots,u_{k})\ast h(u_{0},\ldots,u_{k-m})=\nonumber\\
g_{n}(u_{0}\ast h(u_{1},\ldots,u_{k-m+1}),u_{1}\ast h(u_{2},\ldots
,u_{k-m+2}),\ldots,u_{m-1}\ast h(u_{m},\ldots,u_{k})) \label{scm1}%
\end{gather}

We now define a natural form symmetry concept for the recursive difference
equation (\ref{dek}).

\begin{definition}
\label{fsrec} The function $H=[h_{1},\ldots,h_{m}]$ whose components $h_{j}$
are defined by (\ref{hjd}) is a (recursive)\textit{\ }\textbf{form symmetry}
for Eq.(\ref{dek}). Since the range of $H$ has a lower dimension $m$ than the
dimension $k+1$ of its domain, we say that $H$ is an \textbf{order-reducing}
form symmetry.
\end{definition}

For each solution $\{x_{n}\}$ of (\ref{dek}) according to (\ref{scm}),
(\ref{nplus1}) and (\ref{hjd}) the sequence
\[
t_{n}=x_{n}\ast h(x_{n-1},\ldots,x_{n-k+m-1})
\]
is a solution of (\ref{gnt}) with initial values%

\begin{equation}
t_{-j}=x_{-j}\ast h(x_{-j-1},\ldots,x_{-j-k+m-1}),\quad j=0,\ldots,m-1.
\label{ivfac}%
\end{equation}

Thus, the following pair of lower order equations {}are satisfied:%
\begin{align}
t_{n+1}  &  =g_{n}(t_{n},\ldots,t_{n-m+1}),\label{fac}\\
x_{n+1}  &  =t_{n+1}\ast\lbrack h(x_{n},\ldots,x_{n-k+m})]^{-1} \label{cofac}%
\end{align}
where $-1$ denotes group inversion in $G.$ The pair of equations (\ref{fac})
and (\ref{cofac}) is uncoupled in the sense that (\ref{fac}) is independent of
(\ref{cofac}). Such a pair forms a triangular system as defined in \cite{al} and \cite{js}.

\begin{definition}
Eq.(\ref{fac}) is a \textbf{factor} of Eq.(\ref{dek}) since it is derived from
the semiconjugate factor $\Phi_{n}.$ Eq.(\ref{cofac}) that links the factor to
the original equation is a \textbf{cofactor} of Eq.(\ref{dek}). We refer to
the system of equations (\ref{fac}) and (\ref{cofac}) as a
\textbf{semiconjugate (SC) factorization} of Eq.(\ref{dek}). Note that
\textit{orders }$m$\textit{\ and }$k+1-m$\textit{\ of (\ref{fac}) and
(\ref{cofac}) respectively, add up to the order of (\ref{dek}).}
\end{definition}

\subsection{The factorization theorem}

We are now ready for the following fundamental result, i.e., the equivalence
of Eq.(\ref{dek}) with the system of equations (\ref{fac}) and (\ref{cofac}).

\begin{theorem}
\label{scthm1}\textit{Let }$k\geq1$\textit{, }$1\leq m\leq k$\textit{\ and
suppose that there are functions }$h:G^{k-m+1}\rightarrow G$\textit{\ and
}$g_{n}:G^{m}\rightarrow G$\textit{\ that satisfy equations (\ref{scm}) and
(\ref{hjd}). }

(a) \textit{With the order-reducing form symmetry }%
\begin{equation}
H(u_{0},\ldots,u_{k})=[u_{0}\ast h(u_{1},\ldots,u_{k+1-m}),\ldots,u_{m-1}\ast
h(u_{m},\ldots,u_{k})] \label{Hsc}%
\end{equation}
Eq.(\ref{dek})\textit{\ is equivalent to the SC factorization consisting of
the triangular system of equations (\ref{fac}) and (\ref{cofac}).}

(b) The function $H:G^{k+1}\rightarrow G^{m}$ defined by (\ref{Hsc}) is surjective.

(c) For each $n$, the SC factor map $\Phi_{n}:G^{m}\rightarrow G^{m}$ in
(\ref{sc0}) is the unfolding of Eq.(\ref{fac}). In particular, each $\Phi_{n}$
is of scalar type.
\end{theorem}

\begin{proof}
(a) To show that the SC factorization system consisting of equations
(\ref{fac}) and (\ref{cofac}) is equivalent to Eq.(\ref{dek}) we show that:
(i) each solution $\{x_{n}\}$ of (\ref{dek}) uniquely generates a solution of
(\ref{fac}) and (\ref{cofac}) and conversely (ii) each solution $\{(t_{n}%
,y_{n})\}$ of the system (\ref{fac}) and (\ref{cofac}) correseponds uniquely
to a solution $\{x_{n}\}$ of (\ref{dek}). To establish (i), let $\{x_{n}\}$ be
the unique solution of (\ref{dek}) corresponding to a given set of initial
values $x_{0},\ldots x_{-k}\in G.$ Define the sequence
\begin{equation}
t_{n}=x_{n}\ast h(x_{n-1},\ldots,x_{n-k+m-1}) \label{tn}%
\end{equation}

\noindent for $n\geq-m+1.$ Then for each $n\geq0$ using (\ref{scm1})%
\begin{align*}
x_{n+1}  &  =f_{n}(x_{n},\ldots,x_{n-k})\\
&  =g_{n}(x_{n}\ast h(x_{n-1},\ldots,x_{n-k+m-1}),\ldots,x_{n-m+1}\ast
h(x_{n-m},\ldots,x_{n-k}))\ast\\
&  \hspace{3in}[h(x_{n},\ldots,x_{n-k+m})]^{-1}\\
&  =g_{n}(t_{n},\ldots,t_{n-m+1})\ast\lbrack h(x_{n},\ldots,x_{n-k+m})]^{-1}%
\end{align*}

Therefore, $g_{n}(t_{n},\ldots,t_{n-m+1})=x_{n+1}\ast h(x_{n},\ldots
,x_{n-k+m})=t_{n+1}$ so that $\{t_{n}\}$ is the unique solution of the factor
equation (\ref{fac}) with initial values%
\[
t_{-j}=x_{-j}\ast h(x_{-j-1},\ldots,x_{-j-k+m-1}),\quad j=0,\ldots,m-1.
\]

Further, by (\ref{tn}) for $n\geq0$ we have $x_{n+1}=t_{n+1}\ast\lbrack
h(x_{n},\ldots,x_{n-k+m})]^{-1}$ so that $\{x_{n}\}$ is the unique solution of
the cofactor equation (\ref{cofac}) with initial values $y_{-i}=x_{-i}$ for
$i=0,1,\ldots,k-m$ and $t_{n}$ as obtained above.

To establish (ii), let $\{(t_{n},y_{n})\}$ be a solution of the
factor-cofactor system with initial values
\[
t_{0},\ldots,t_{-m+1},y_{-m},\ldots y_{-k}\in G.
\]

Note that these numbers determine $y_{-m+1},\ldots,y_{0}$ through the cofactor
equation%
\begin{equation}
y_{-j}=t_{-j}\ast\lbrack h(y_{-j-1},\ldots,y_{-j-1-k+m})]^{-1},\quad
j=0,\ldots,m-1. \label{y-n}%
\end{equation}

Now using (\ref{scm1}) for $n\geq0$ we obtain%
\begin{align*}
y_{n+1}  &  =t_{n+1}\ast\lbrack h(y_{n},\ldots,y_{n-k+m})]^{-1}\\
&  =g_{n}(t_{n},\ldots,t_{n-m+1})\ast\lbrack h(y_{n},\ldots,y_{n-k+m})]^{-1}\\
&  =g_{n}(y_{n}\ast h(y_{n-1},\ldots,y_{n-k+m-1}),\ldots,y_{n-m+1}\ast
h(y_{n-m},\ldots,y_{n-k}))\ast\\
&  \hspace{3in}[h(y_{n},\ldots,y_{n-k+m})]^{-1}\\
&  =f_{n}(y_{n},\ldots,y_{n-k})
\end{align*}

Thus $\{y_{n}\}$ is the unique solution of Eq.(\ref{dek}) that is generated by
the initial values (\ref{y-n}) and $y_{-m},\ldots y_{-k}.$ This completes the
proof of (a).

(b) Choose an arbitrary point $[v_{1},\ldots,v_{m}]\in G^{m}$ and set
\[
u_{m-1}=v_{m}\ast h(u_{m},u_{m+1}\ldots,u_{k})^{-1}%
\]
where $u_{m}=u_{m+1}=\ldots u_{k}=\bar{u}$ where $\bar{u}$ is a fixed element
of $G,$ e.g., the identity.\ Then
\begin{align*}
v_{m}  &  =u_{m-1}\ast h(\bar{u},\bar{u}\ldots,\bar{u})\\
&  =u_{m-1}\ast h(u_{m},u_{m+1}\ldots,u_{k})\\
&  =h_{m}(u_{0},\ldots,u_{k})\\
&  =h_{m}(u_{0},\ldots,u_{m-2},v_{m}\ast h(\bar{u},\bar{u}\ldots,\bar{u}%
)^{-1},\bar{u}\ldots,\bar{u}).
\end{align*}
for any choice of points $u_{0},\ldots,u_{m-2}\in G.$ Similarly, define
\[
u_{m-2}=v_{m-1}\ast h(u_{m-1},u_{m}\ldots,u_{k-1})^{-1}%
\]
so as to get
\begin{align*}
v_{m-1}  &  =u_{m-2}\ast h(u_{m-1},u_{m}\ldots,u_{k-1})\\
&  =h_{m-1}(u_{0},\ldots,u_{k})\\
&  =h_{m-1}(u_{0},\ldots,u_{m-3},v_{m-1}\ast h(u_{m-1},\bar{u}\ldots,\bar
{u})^{-1},u_{m-1},\bar{u}\ldots,\bar{u})
\end{align*}
for any choice of $u_{0},\ldots,u_{m-3}\in G.$ Continuing in this way,
induction leads to selection of $u_{m-1},\ldots,u_{0}$ such that%
\[
v_{j}=h_{j}(u_{0},\ldots,u_{m-1},\bar{u}\ldots,\bar{u}),\quad j=1,\ldots,m.
\]
Therefore, $H(u_{0},\ldots,u_{m-1},\bar{u}\ldots,\bar{u})=[v_{1},\ldots
,v_{m}]$ and it follows that $H$ is onto $G^{m}.$

(c) It is necessary to prove that each coordinate function $\phi_{j,n}$ is the
projection into coordinate $j-1$ for $j>1.$ Suppose that the maps $h_{j}$ are
given by (\ref{hjd}). For $j=2$ (\ref{scj}) gives%
\begin{align*}
\phi_{2,n}(h_{1}(u_{0},\ldots,u_{k}),\ldots,h_{m}(u_{0},\ldots,u_{k}))  &
=h_{2}(f_{n}(u_{0},\ldots,u_{k}),u_{0},\ldots,u_{k-1})\\
&  =u_{0}\ast h(u_{1},u_{2}\ldots,u_{k-m+1})\\
&  =h_{1}(u_{0},\ldots,u_{k}).
\end{align*}
Therefore, $\phi_{2,n}$ projects into coordinate 1. Generally, for $j\geq2$ we
have%
\begin{align*}
\phi_{j,n}(h_{1}(u_{0},\ldots,u_{k}),\ldots,h_{m}(u_{0},\ldots,u_{k}))  &
=h_{j}(f_{n}(u_{0},\ldots,u_{k}),u_{0},\ldots,u_{k-1})\\
&  =u_{j-2}\ast h(u_{j-1},u_{j}\ldots,u_{j+k-m-1})\\
&  =h_{j-1}(u_{0},\ldots,u_{k}).
\end{align*}
Therefore, for each $n$ and for every $(t_{1},\ldots,t_{m})\in H(G^{k+1})$ we
have%
\[
\Phi_{n}(t_{1},\ldots,t_{m})=[g_{n}(t_{1},\ldots,t_{m}),t_{1},\ldots,t_{m-1}]
\]
i.e., $\Phi_{n}|_{H(G^{k+1})}$ is of scalar type. Since by Part (b)
$H(G^{k+1})=G^{m}$ it follows that $\Phi_{n}$ is of scalar type.
\end{proof}

We point out that the SC factorization in Theorem \ref{scthm1}(a) does not
require the determination of $\phi_{j,n}$ for $j\geq2.$ However, as seen in
Parts (b) and (c) of the theorem the rest of the picture fits together properly.

We discuss several applications of Theorem \ref{scthm1} in later sections
below. The following example gives an application in finite settings.

\begin{example}
\label{z2c}Let $G$ be any nontrivial abelian group (e.g., $\mathbb{Z}_{2}$ in
Example \ref{z2b} viewed as an additive group) and consider the following
difference equation of order two%
\begin{equation}
x_{n+1}=x_{n-1}+a,\quad a,x_{0},x_{-1}\in G,\ a\not =0.\label{z2cde}%
\end{equation}
The unfolding of Eq.(\ref{z2cde}) is the map $F(u,v)=[v+a,u].$ Note that the
function $H(u,v)=u+v$ is of type $u+h(v)$ with $h$ being the identity function
on $G$ ($H$ is the same as $H_{6}$ in Example \ref{z2b}). With this $H$ we
have
\[
H(F(u,v))=u+v+a=\phi(H(u,v))
\]
where $\phi:G\rightarrow G$ is defined as $\phi(t)=t+a$. Therefore,
$F\trianglerighteq\phi.$ Using Theorem \ref{scthm1} we obtain the SC
factorization of Eq.(\ref{z2cde}) as
\begin{align*}
t_{n+1} &  =\phi(t_{n})=t_{n}+a,\ t_{0}=H(x_{0},x_{-1})=x_{0}+x_{-1}\\
x_{n+1} &  =t_{n+1}-x_{n}=t_{n+1}+x_{n}.
\end{align*}

\end{example}

\subsection{Reduction types and chains}

Theorem \ref{scthm1} leads to a natural classification scheme for order
reduction which we discuss in this section. We begin with a definition.

\begin{definition}
The SC factorization of Eq.(\ref{dek}) into (\ref{fac}) and (\ref{cofac})
gives a \textbf{type-(}$m,k+1-m$\textbf{) order reduction} (or just
\textbf{type-(}$m,k+1-m$\textbf{) reduction}) for (\ref{dek}). We also say
that (\ref{dek}) is a type-($m,k+1-m$) equation in this case.
\end{definition}

A second-order difference equation ($k=1$) can have only the type-(1,1) order
reduction into two first order equations. A third-order equation has two order
reduction types, namely (1,2) and (2,1), a fourth order equation has three
order reduction types (1,2), (2,2) and (2,1) and so on. Of the $k$ possible
order reduction types
\[
(1,k),\ (2,k-1),\cdots,(k-1,2),\ (k,1)
\]
for an equation of order $k+1$ the two extreme ones, namely $(1,k)$ and
$(k,1)$ have the extra appeal of having an equation of order 1 as either a
factor or a cofactor.

Eq.(\ref{dek}) may admit repeated reductions of order through its factor
equation, its cofactor equation or both as follows:%
\[
\text{Eq.(\ref{dek})}\rightarrow\left\{
\begin{array}
[c]{ll}%
\underset{\text{factor equation}}{\underbrace{t_{n+1}=g_{n}(t_{n}%
,\ldots,t_{n-m+1})}}\quad\rightarrow & \left\{
\begin{array}
[c]{l}%
\text{factor\ }\rightarrow\cdots\\
\text{cofactor }\rightarrow\cdots
\end{array}
\right. \\
\underset{\text{cofactor equation}}{\underbrace{x_{n+1}=t_{n+1}\ast
h(x_{n},\ldots,x_{n-k+m})^{-1}}}\,\rightarrow & \left\{
\begin{array}
[c]{l}%
\text{factor }\rightarrow\cdots\\
\text{cofactor }\rightarrow\cdots
\end{array}
\right.
\end{array}
\right.
\]

In the above binary tree structure, we call each branch a \textit{reduction
chain}. If a reduction chain consists only of factor (or cofactor) equations
then it is a \textit{factor }(or\textit{\ cofactor}) \textit{chain}. In
particular, we show later (Corollary \ref{lincor} below) that a linear
difference equation of order $k+1$ has a full cofactor chain leading to a
system of $k+1$ linear first order equations, as discussed earlier in the
introduction. The next example exhibits a full factor chain.

\begin{example}
\label{chains}Consider the following reduction chain:%
\begin{equation}
x_{n+1}=x_{n}+\frac{a(x_{n}-x_{n-1})^{2}}{x_{n-1}-x_{n-2}}\rightarrow\left\{
\begin{array}
[c]{ll}%
t_{n+1}=\dfrac{at_{n}^{2}}{t_{n-1}}\ \rightarrow & \left\{
\begin{array}
[c]{l}%
s_{n+1}=as_{n}\\
t_{n+1}=s_{n+1}t_{n}%
\end{array}
\right. \\
x_{n+1}=t_{n+1}+x_{n} &
\end{array}
\right.  \label{cscsys1}%
\end{equation}
The 3rd order equation is reduced by substituting $t_{n}=x_{n}-x_{n-1}$, then
the 2nd order factor equation is reduced by the substitution $s_{n}%
=t_{n}/t_{n-1}$ (see Definition \ref{hd1def} below and the comments that
follow it). The factor chain in this example has length three as follows:%
\[
x_{n+1}=x_{n}+\frac{a(x_{n}-x_{n-1})^{2}}{x_{n-1}-x_{n-2}}\rightarrow
t_{n+1}=\dfrac{at_{n}^{2}}{t_{n-1}}\rightarrow s_{n+1}=as_{n}.
\]

\end{example}

Since a first order difference equation does not have order-reducing form
symmetries the SC factorization process of an equation of order $k+1$ must
stop in at most $k$ steps. The result then is a system of $k+1$\ first order
equations that mark the ends of factor/cofactor chains. In (\ref{cscsys1}) the
three first order equations that mark the ends of factor/cofactor chains can
be arranged as follows:%
\begin{align}
s_{n+1}  &  =as_{n}\nonumber\\
t_{n+1}  &  =as_{n}t_{n}\label{triangsys1}\\
x_{n+1}  &  =as_{n}t_{n}+x_{n}.\nonumber
\end{align}

Since each equation depends on the variables in the equations above it,
(\ref{triangsys1}) is a \textquotedblleft triangular system.\textquotedblright%
\ In principle, a triangular system can be solved by solving the top-most
equation and then using that solution to solve the next equation and so on. In
the case of (\ref{triangsys1}) since the equations are linear (including the
nonautonomous and nonhomogenous versions) an explicit formula for solutions of
the system can be determined; in particular, the third order equation
(\ref{cscsys1}) is integrable. Whether a complete SC factorization as a
triangular system exists for each higher order difference equation of type
(\ref{dek}) is a difficult question that is equivalent to the apparently more
basic problem of existence of semiconjugate relations for (\ref{dek}).

\section{Reduction of order of difference equations}

In this section we use the methods of the previous section to obtain
reductions in orders of various classes of difference equations.

\subsection{Equations with type-$(k,1)$ reductions}

For type-$(k,1)$ reductions of Eq.(\ref{dek}) $m=k$. Therefore, the function
$h:G\rightarrow G$ in (\ref{hjd}) is of one variable and yields the form symmetry%

\begin{equation}
H(u_{0},\ldots,u_{k})=[u_{0}\ast h(u_{1}),u_{1}\ast h(u_{2})\ldots,u_{k-1}\ast
h(u_{k})]. \label{km}%
\end{equation}

Theorem \ref{scthm1} gives the SC factorization as the pair%

\begin{align}
t_{n+1}  &  =g_{n}(t_{n},t_{n-1},\ldots,t_{n-k+1})\label{kmscf}\\
x_{n+1}  &  =t_{n+1}\ast\lbrack h(x_{n})]^{-1}. \label{kmsccf}%
\end{align}

The functions $g_{n}:G^{k}\rightarrow G$ are determined by the given functions
$f_{n}$ in (\ref{dek}) as in the previous sections. From the semiconjugate
relation it follows that a type-($k,1$) reduction with a form symmetry of type
(\ref{km}) exists if and only if there are functions $g_{n}$ such that
\begin{equation}
g_{n}(u_{0}\ast h(u_{1}),u_{1}\ast h(u_{2}),\ldots,u_{k-1}\ast h(u_{k}%
))=f_{n}(u_{0},u_{1},\ldots,u_{k})\ast h(u_{0}). \label{fsk1phi}%
\end{equation}

\medskip

\noindent \textbf{Remark 1.} 
The order of the factor equation (\ref{kmscf}) is one less than the order of (\ref{dek}).
If a solution $t_{n+1}$ of (\ref{kmscf}) is known then by the factorization theorem the 
corresponding solution $x_{n}$ of (\ref{dek}) is obtained by solving the cofactor equation
(\ref{kmsccf}), which has order one. For this reason we may refer to Eq.(\ref{kmscf}) as an
\textit{order reduction} for Eq.(\ref{dek}).

\subsubsection{Invertibility criterion}

The next result from \cite{hsinvcrt} gives a necessary and sufficient
condition for the existence of $\phi_{n}$ when the function $h$ above is
invertible. We repeat the proof for the reader's convenience.

\textit{Notation:} the symbol $h^{-1}(\cdot)$ denotes the inverse of $h$ as a
function and the symbol $h(\cdot)^{-1}$ denotes group inversion.

\begin{theorem}
\label{hinvfs}(Invertibility criterion) Let $h:G\rightarrow G$ be a bijection.
For $u_{0},v_{1},\ldots,v_{k}\in G$ let $\zeta_{0}=u_{0}$ and define
\begin{equation}
\zeta_{j}=h^{-1}(\zeta_{j-1}^{-1}\ast v_{j}),\quad j=1,\ldots,k. \label{hinvz}%
\end{equation}
Then Eq.(\ref{dek}) has the form symmetry (\ref{km}) with SC factors $\phi
_{n}$ satisfying (\ref{fsk1phi}) if and only if the quantity
\begin{equation}
f_{n}(u_{0},\zeta_{1},\ldots,\zeta_{k})\ast h(u_{0}) \label{hinvcrit}%
\end{equation}
is independent of $u_{0}$ for every $n$.
\end{theorem}

\begin{proof}
First assume that the quantity in (\ref{hinvcrit}) is independent of $u_{0}$
for all $v_{1},\ldots,v_{k}$ so that the function
\begin{equation}
g_{n}(v_{1},\ldots,v_{k})=f_{n}(u_{0},\zeta_{1},\ldots,\zeta_{k})\ast h(u_{0})
\label{fsidcond0}%
\end{equation}
is well defined. Next, if $H$ is given by (\ref{km}) then for all $u_{0}%
,u_{1},\ldots,u_{k}$
\[
g_{n}(H(u_{0},u_{1},\ldots,u_{k}))=g_{n}(u_{0}\ast h(u_{1}),u_{1}\ast
h(u_{2}),\ldots,u_{k-1}\ast h(u_{k})).
\]
Define
\begin{equation}
v_{j}=u_{j-1}\ast h(u_{j}),\quad j=1,\ldots,k. \label{fsid1}%
\end{equation}
Then by (\ref{hinvz})
\[
\zeta_{1}=h^{-1}(u_{0}^{-1}\ast v_{1})=h^{-1}(u_{0}^{-1}\ast u_{0}\ast
h(u_{1}))=u_{1}.
\]
In fact $\zeta_{j}=u_{j}$ for every $j,$ for if by way of induction $\zeta
_{l}=u_{l}$ for $1\leq l<j$ then
\[
\zeta_{j}=h^{-1}(\zeta_{j-1}^{-1}\ast v_{j})=h^{-1}(u_{j-1}^{-1}\ast
u_{j-1}\ast h(u_{j}))=u_{j}.
\]
Now by (\ref{fsidcond0})
\[
g_{n}(H(u_{0},u_{1},\ldots,u_{k}))=f_{n}(u_{0},\ldots,u_{k})\ast h(u_{0})
\]
so if $F_{n}$ and $\Phi_{n}$ are the unfoldings of $f_{n}$ and $g_{n}$
respectively, then
\begin{align*}
H(F_{n}(u_{0},\ldots,u_{k}))  &  =[f_{n}(u_{0},\ldots,u_{k})\ast
h(u_{0}),u_{0}\ast h(u_{1}),\ldots,u_{k-2}\ast h(u_{k-1})]\\
&  =[g_{n}(H(u_{0},u_{1},\ldots,u_{k})),u_{0}\ast h(u_{1}),\ldots,u_{k-2}\ast
h(u_{k-1})]\\
&  =\Phi_{n}(H(u_{0},\ldots,u_{k}))
\end{align*}
Finally, to show that $H$ is a semiconjugate link for Eq.(\ref{dek}) we show
that $H$ is onto $G^{k}.$ Let $(v_{1},\ldots,v_{k})\in G^{k}$ and let $u_{k}$
be any element of $G,$ e.g., the identity. Define
\[
u_{j-1}=v_{j}\ast\lbrack h(u_{j})]^{-1},\quad j=k,k-1,\ldots,1
\]

\noindent and note that
\begin{align*}
H(u_{0},\ldots,u_{k})  &  =[u_{0}\ast h(u_{1}),u_{1}\ast h(u_{2}%
)\ldots,u_{k-1}\ast h(u_{k})]\\
&  =[v_{1},v_{2},\ldots,v_{k}]
\end{align*}

\noindent i.e., $H$ is onto $G^{k}$ as claimed.

Conversely, if $H$ as given by (\ref{km}) is a form symmetry then the
semiconjugate relation implies that there are functions $g_{n}$ such that%
\begin{equation}
f_{n}(u_{0},\ldots,u_{k})\ast h(u_{0})=g_{n}(u_{0}\ast h(u_{1}),\ldots
,u_{k-1}\ast h(u_{k})). \label{fsidcond1}%
\end{equation}
For every $v_{1},\ldots,v_{k}$ in $G$ and with $\zeta_{j}$ as defined in
(\ref{hinvz}),%
\begin{align*}
f_{n}(u_{0},\zeta_{1},\ldots,\zeta_{k})\ast u_{0}  &  =g_{n}(u_{0}\ast
h(\zeta_{1}),\zeta_{1}\ast h(\zeta_{2}),\ldots,\zeta_{k-1}\ast h(\zeta_{k}))\\
&  =g_{n}(v_{1},\ldots,v_{k})
\end{align*}
which is clearly independent of $u_{0}.$
\end{proof}

It is worth noting that (\ref{hinvz}) is a backwards version of the cofactor
equation (\ref{kmsccf}) that is obtained by solving it for $x_{n}$ instead of
$x_{n+1}.$ To do this we required $h$ to be invertible. Of course, in
(\ref{hinvz}) it is necessary to iterate only $k$ times.

\subsubsection{Identity and inversion form symmetries}

We now examine two of the simplest possible form symmetries within the context
of the previous section that are based on the scalar maps%
\[
h_{1}(u)=u\quad\text{and}\quad h_{-1}(u)=u^{-1}.
\]

Each of $h_{\pm1}$ is invertible and equals itself (self-inverse maps).

\begin{definition}
The form symmetry of type (\ref{km}) that is generated by $h_{1}$ is the
\textbf{identity form symmetry} and that generated by $h_{-1}$ is the\textbf{
inversion form symmetry}.
\end{definition}

The next result is an immediate consequence of Theorem \ref{hinvfs}.

\begin{corollary}
\label{fsidthm}(a) For every $u_{0},v_{1},\ldots,v_{k}\in G$ let $\zeta
_{0}=u_{0}$ and define $\zeta_{j}=\zeta_{j-1}^{-1}\ast v_{j}$ for
$j=1,\ldots,k.$ Then Eq.(\ref{dek}) has the identity form symmetry if and only
if the quantity
\begin{equation}
f_{n}(u_{0},\zeta_{1},\ldots,\zeta_{k})\ast u_{0} \label{fsidcond}%
\end{equation}
is independent of $u_{0}$ for every $n$.

(b) For every $u_{0},v_{1},\ldots,v_{k}\in G$ let $\zeta_{0}=u_{0}$ and define
$\zeta_{j}=v_{j}^{-1}\ast\zeta_{j-1}$ for $j=1,\ldots,k.$ Then Eq.(\ref{dek})
has the inversion form symmetry if and only if the quantity
\begin{equation}
f_{n}(u_{0},\zeta_{1},\ldots,\zeta_{k})\ast u_{0}^{-1} \label{fsinvcond}%
\end{equation}
is independent of $u_{0}$ for every $n$.
\end{corollary}

For the inversion form symmetry we can improve on Corollary \ref{fsidthm}(b)
considerably. In fact, there is a simple characterization of functions $f_{n}$
which satisfy (\ref{fsinvcond}).

\begin{definition}
\label{hd1def}Equation (\ref{dek}) is \textbf{homogeneous of degree one (HD1)}
if for every $n=0,1,2,\ldots$ the functions $f_{n}$ are all homogeneous of
degree one relative to the group $G,$ i.e.,%
\begin{gather*}
f_{n}(u_{0}\ast t,\ldots,u_{k}\ast t)=f_{n}(u_{0},\ldots,u_{k})\ast t\text{
}\\
\text{for all }t,u_{i}\in G,\ i=0,\ldots,k,\ \text{and all }n\geq0.
\end{gather*}

\end{definition}

Note that the third order equation in Example \ref{chains} is HD1 with respect
to the additive group of real numbers while its second order factor is HD1
with respect to the multiplicative group of nonzero real numbers.

There is an abundance of HD1 functions on groups. For instance, if $G$\ is a
nontrivial group, $k$ a positive integer and $g:G^{k}\rightarrow G$\ is any
given function, then it is easy to verify that the mappings $\bar{g}%
_{j}:G^{k+1}\rightarrow G$ defined by\
\[
\bar{g}_{j}(u_{0},\ldots,u_{k})=g(u_{0}\ast u_{1}^{-1},u_{1}\ast u_{2}%
^{-1},\ldots,u_{k-1}\ast u_{k}^{-1})\ast u_{j}%
\]
are HD1 functions for every $j=0,1,\ldots,k$. Further, if $g_{1},g_{2}%
:$\ $G^{k+1}\rightarrow G$ and $f:G^{2}\rightarrow G$ are HD1 then so is the
composition%
\[
f(g_{1}(u_{0},\ldots,u_{k}),g_{2}(u_{0},\ldots,u_{k})).
\]

The next result was proved by direct arguments in \cite{hd1}; also see \cite{hsarx}. Here we give a different proof using Corollary \ref{fsidthm}(b) and Theorem \ref{scthm1}.

\begin{corollary}
\label{hd1thm} \noindent\textit{Eq.(\ref{dek}) has the inversion form symmetry
if and only if }$f_{n}$\textit{\ is HD1 relative to }$G$\textit{\ for all }%
$n$. \textit{In this case, (\ref{dek}) has a type-(}$k,1$\textit{)
order-reduction with the SC factorization}
\begin{align}
t_{n+1}  &  =f_{n}(1,t_{n}^{-1},(t_{n}\ast t_{n-1})^{-1},\ldots,(t_{n}\ast
t_{n-1}\ast\cdots\ast t_{n-k+1})^{-1})\label{1aa}\\
x_{n+1}  &  =t_{n+1}\ast x_{n}. \label{1ab}%
\end{align}

\end{corollary}

\begin{proof}
First if $\zeta_{j}=$ $v_{j}^{-1}\ast\zeta_{j-1}$ as in Corollary
\ref{fsidthm}(b) then by straightforward iteration
\begin{equation}
\zeta_{j}=v_{j}^{-1}\ast\cdots\ast v_{1}^{-1}\ast u_{0},\quad j=1,\ldots,m.
\label{zetaj}%
\end{equation}

Now if $f_{n}$ is HD1 for every $n$ then%
\[
f_{n}(u_{0},\zeta_{1},\ldots,\zeta_{k})\ast u_{0}^{-1}=f_{n}(1,\zeta_{1}\ast
u_{0}^{-1},\ldots,\zeta_{k}\ast u_{0}^{-1})
\]

\noindent which by (\ref{zetaj}) is independent of $u_{0}.$ Thus by Corollary
\ref{fsidthm}(b) Eq.(\ref{dek}) has the inversion form symmetry.

Conversely, assume that (\ref{dek}) has the inversion form symmetry. Then by
Corollary \ref{fsidthm}(b) for every $u_{0},v_{1},\ldots,v_{k}\in G$ the
quantity in (\ref{fsinvcond}) is independent of $u_{0}.$ Thus there are
functions $\phi_{n}$ where
\begin{align}
f_{n}(u_{0},\zeta_{1},\ldots,\zeta_{k})\ast u_{0}^{-1}  &  =\phi_{n}%
(v_{1},\ldots,v_{k})\label{hd1a}\\
&  =\phi_{n}(\zeta_{0}\ast\zeta_{1}^{-1},\ldots,\zeta_{k-1}\ast\zeta_{k}%
^{-1}).\nonumber
\end{align}

Note that (\ref{hd1a}) holds for arbitrary values of $u_{0},\zeta_{1}%
,\ldots,\zeta_{k}$ since $v_{1},\ldots,v_{k}$ are arbitrary. Thus for all
$t,s_{0},\ldots,s_{k}\in G$ and all $n$,%
\begin{align*}
f_{n}(s_{0}\ast t,\ldots,s_{k}\ast t)  &  =\phi_{n}((s_{0}\ast t)\ast
(s_{1}\ast t)^{-1},\ldots,(s_{k-1}\ast t)\ast(s_{k}\ast t)^{-1})\ast(s_{0}\ast
t)\\
&  =[\phi_{n}(s_{0}\ast s_{1}^{-1},\ldots,s_{k-1}\ast s_{k}^{-1})\ast
s_{0}]\ast t\\
&  =f_{n}(s_{0},\ldots,s_{k})\ast t.
\end{align*}
It follows that $f_{n}$ is HD1 relative to $G$ for all $n$.

Finally, the inversion form symmetry
\[
H(u_{0},\ldots,u_{k})=[u_{0}\ast u_{1}^{-1},\ldots,u_{k-1}\ast u_{k}^{-1}]
\]

\noindent yields a semiconjugate relation that changes variables to
$t_{n}=x_{n}/x_{n-1}$ and the SC factorization (\ref{1aa}) and (\ref{1ab}) is
obtained by Theorem \ref{scthm1}.
\end{proof}

\begin{example}
\label{ord3ratid}Consider the 3rd order difference equation%
\begin{equation}
x_{n+1}=\frac{x_{n}x_{n-1}}{a_{n}x_{n}+b_{n}x_{n-2}},\quad a_{n},b_{n}>0.
\label{ord3idfs}%
\end{equation}
Here for all $n$, $f_{n}(u_{0},u_{1},u_{2})=u_{0}u_{1}/(a_{n}u_{0}+b_{n}%
u_{2})$ is HD1 on $(0,\infty)$ under ordinary multiplication. By Corollary
\ref{hd1thm} Eq.(\ref{ord3idfs}) has the following SC factorization%
\begin{align}
t_{n+1}  &  =\frac{t_{n-1}}{a_{n}t_{n}t_{n-1}+b_{n}}\label{ord3ratidfac}\\
x_{n+1}  &  =t_{n+1}x_{n}.\nonumber
\end{align}
Next, consider the factor equation (\ref{ord3ratidfac}) again on $(0,\infty)$
under ordinary multiplication. This equation has order two and represents the 
order reduced form of Eq.(\ref{ord3idfs}) in the sense of Remark 1. Note that
(\ref{ord3ratidfac}) is not HD1; however,
substituting $\zeta_{1}=v_{1}/u_{0}$ in the functions $\phi_{n}(u_{0}%
,u_{1})=u_{1}/(a_{n}u_{0}u_{1}+b_{n})$ gives%
\[
\phi_{n}\left(  u_{0},\frac{v_{1}}{u_{0}}\right)  u_{0}=\frac{v_{1}/u_{0}%
}{a_{n}u_{0}v_{1}/u_{0}+b_{n}}u_{0}=\frac{v_{1}}{a_{n}v_{1}+b_{n}}%
\]
which is independent of $u_{0}.$ Thus Corollary \ref{fsidthm} implies that
Eq.(\ref{ord3ratidfac}) has the identity form symmetry. Its SC factorization
is easily found using Theorem \ref{scthm1} upon substituting $r_{n}%
=t_{n}t_{n-1}$ to get%
\begin{align}
r_{n+1}  &  =\frac{r_{n}}{a_{n}r_{n}+b_{n}}\label{ord2rathd1fac}\\
t_{n+1}  &  =\frac{r_{n+1}}{t_{n}}.\nonumber
\end{align}
The further substitution $s_{n}=1/r_{n}$ in (\ref{ord2rathd1fac}) produces a
linear equation $s_{n+1}=a_{n}+b_{n}s_{n}$ and yields a full factorization of
Eq.(\ref{ord3idfs}) into the following triangular system of first order
equations%
\begin{align*}
s_{n+1}  &  =a_{n}+b_{n}s_{n}\\
t_{n+1}  &  =\frac{1}{s_{n+1}t_{n}}\\
x_{n+1}  &  =t_{n+1}x_{n}.
\end{align*}

\end{example}

The factorizations in the preceding example were used in analyzing the global
behavior of Eq.(\ref{ord3idfs}) in \cite{ord23}. SC factorization of HD1
equations has also been used in \cite{dkmos1} and \cite{ks} to study the
global behavior of difference equations of order two.

\begin{corollary}
\label{rick}Let $\mathcal{F}$ be a nontrivial algebraic field and
$a_{1},\ldots,a_{k}\in\mathcal{F}$ with $a_{k}\not =0.$ Each linear
homogeneous difference equation%
\[
x_{n+1}=a_{1}x_{n}+\cdots+a_{k}x_{n-k}%
\]
of order $k+1$ is HD1 on the mulitplicative group $\mathcal{F}\backslash
\{0\}$. Thus it admits a type-$(k,1)$ reduction with factor equation%
\[
t_{n+1}=a_{1}+\frac{a_{2}}{t_{n}}+\cdots+\frac{a_{k}}{t_{n}t_{n-1}\cdots
t_{n-k+1}}.
\]

\end{corollary}

The factor equation in Corollary \ref{rick} is known as the Riccati difference
equation of order $k;$ see \cite{ric2}.

\subsubsection{The linear form symmetry\label{linfs1}}

In practice, the one-variable function $h$ is often defined using structures
that are more complex than a group. In particular, on a field $\mathcal{F}$
the function%
\[
h(u)=-\alpha u
\]
where $\alpha$ is a fixed nonzero element of the field defines a form symmetry%
\begin{equation}
H(u_{0},u_{1},\ldots,u_{k})=[u_{0}-\alpha u_{1},u_{1}-\alpha u_{2}%
,\ldots,u_{k-1}-\alpha u_{k}]. \label{fslin}%
\end{equation}

For convenience, we represent the field operations here by the ordinary
addition and multiplication symbols. We call $H$ defined in (\ref{fslin}) the
\textit{linear form symmetry}. Relative to the additive group of a field, the
linear form symmetry generalizes both the identity form symmetry ($\alpha=-1$)
and the inversion form symmetry ($\alpha=1$).

The linear form symmetry is characterized by a change of variables to
\[
t_{n}=x_{n}-\alpha x_{n-1}.
\]

The factor and cofactor equations for the SC factorization of (\ref{dek}) with
the linear form symmetry are given by (\ref{kmscf}) and (\ref{kmsccf}) as
\begin{align}
t_{n+1}  &  =\phi_{n}(t_{n},\ldots,t_{n-k+1}),\label{linfsfac}\\
x_{n+1}  &  =t_{n+1}+\alpha x_{n} \label{linfscofac}%
\end{align}

For a given sequence $\{t_{n}\}$ in $\mathcal{F}$ the general solution of the
linear cofactor equation (\ref{linfscofac}) is easily found to be%
\begin{equation}
x_{n}=\alpha^{n}x_{0}+\sum_{j=1}^{n}\alpha^{n-j}t_{j},\quad n\geq1.
\label{lincofacsoln}%
\end{equation}

This equation gives a solution of Eq.(\ref{dek}) if $\{t_{n}\}$ is a solution
of (\ref{linfsfac}) with the form symmetry (\ref{fslin}). Difference equations
possessing the linear form symmetry have been studied in \cite{dkmos2} and
\cite{kyoto} where the SC factorization above has been used to reduce
equations of order 2 to pairs of equations of order 1.

The following corollary of Theorem \ref{hinvfs} gives a condition for
verifying whether Eq.(\ref{dek}) has the linear form symmetry. For easier
reading we denote the multiplicative field inversion by the reciprocal notation.

\begin{corollary}
\label{linfsthm}For arbitrary $u_{0},v_{1},\ldots,v_{k}$ in a field
$\mathcal{F}$ define $\zeta_{0}=u_{0}$ and
\begin{equation}
\zeta_{j}=\frac{u_{0}}{\alpha^{j}}-\sum_{i=1}^{j}\frac{v_{i}}{\alpha^{j-i+1}%
},\quad j=1,\ldots,k. \label{linfsa}%
\end{equation}
Then Eq.(\ref{dek}) has the linear form symmetry (\ref{fslin}) if and only if
the quantity%
\[
f_{n}(u_{0},\zeta_{1},\ldots,\zeta_{k})-\alpha u_{0}%
\]
is independent of $u_{0}.$
\end{corollary}

Note that (\ref{linfsa}) defines $\zeta_{j}$ in Corollary \ref{linfsthm}
explicitly rather than recursively. It is obtained from the recursive
definition (\ref{hinvz}) by a simple calculation; since%
\[
\zeta_{j}=h^{-1}(\zeta_{j-1}^{-1}\ast v_{j})=-\frac{1}{\alpha}(-\zeta
_{j-1}+v_{j})=\frac{\zeta_{j-1}-v_{j}}{\alpha}%
\]
equality (\ref{linfsa}) can be established by direct iteration. The next
example illustrates Corollary \ref{linfsthm}.

\begin{example}
\label{linexrat}Let $\mathcal{F}$ be a field and $g_{n}:\mathcal{F}%
\rightarrow\mathcal{F}$ be a given sequence of functions. Consider the
difference equation%
\begin{equation}
x_{n+1}=ax_{n-j}+g_{n}(x_{n}+cx_{n-k}) \label{linfsex}%
\end{equation}
where $c\not =0$, $0\leq j\leq k$ and we define $f_{n}(u_{0},u_{1}%
,\ldots,u_{k})=au_{j}+g_{n}(u_{0}+cu_{k}).$ Using (\ref{linfsa}) we obtain%
\begin{align*}
&  f_{n}\left(  u_{0},\frac{u_{0}-v_{1}}{\alpha},\ldots,\frac{u_{0}}%
{\alpha^{k}}-\sum_{i=1}^{k}\frac{v_{i}}{\alpha^{k-i+1}}\right)  -\alpha
u_{0}\\
&  =\left(  -\alpha+\frac{a}{\alpha\,^{j}}\right)  u_{0}-\sum_{i=1}^{j}%
\frac{av_{i}}{\alpha^{j-i+1}}+g_{n}\left(  u_{0}\left[  1+\frac{c}{\alpha^{k}%
}\right]  -\sum_{i=1}^{k}\frac{cv_{i}}{\alpha^{k-i+1}}\right)  .
\end{align*}
If there is $\alpha\in\mathcal{F}$ such that
\begin{equation}
-\alpha+\frac{a}{\alpha\,^{j}}=0\text{ or }a=\alpha^{j+1}\quad\text{and\quad
}c=-\alpha^{k} \label{coeffs}%
\end{equation}
then by Corollary \ref{linfsthm} Eq.(\ref{linfsex}) has the linear form
symmetry with a SC factorization%
\begin{align}
t_{n+1}  &  =-\sum_{i=1}^{j}\alpha^{i}t_{n-i+1}+g_{n}\left(  \sum_{i=1}%
^{k}\alpha^{i-1}t_{n-i+1}\right)  ,\label{faclinexrat}\\
x_{n+1}  &  =t_{n+1}+\alpha x_{n}=\alpha^{n}x_{0}+\sum_{j=1}^{n}\alpha
^{n-j}t_{j}.\nonumber
\end{align}

\end{example}

\noindent\textbf{Remark}. The above example in particular shows that the
nature of the field $\mathcal{F}$ is important in obtaining SC factorizations.
For example, the difference equation%
\[
x_{n+1}=x_{n}+g_{n}(x_{n}+x_{n-k})
\]
satisfies (\ref{coeffs}) on $\mathbb{Z}_{2}$ but not on $\mathbb{C}$.

\subsection{Equations with type-$(1,k)$ reductions}

In this section we obtain type-$(1,k)$ reductions for a class of equations on
the field $\mathbb{C}$ of complex numbers that includes all linear
nonhomogeneous difference equations with constant coefficients as well as some
interesting nonlinear equations.

For a type-$(1,k)$ reduction $m=1.$ Therefore, $h:G^{k}\rightarrow G$ and the
form symmetry has the scalar form
\begin{equation}
H(u_{0},\ldots,u_{k})=u_{0}\ast h(u_{1},\ldots,u_{k}). \label{mk}%
\end{equation}

This form symmetry gives the SC factorization%
\begin{align}
t_{n+1}  &  =g_{n}(t_{n})\label{mkscf}\\
x_{n+1}  &  =t_{n+1}\ast\lbrack h(x_{n},\ldots,x_{n-k+1})]^{-1}.
\label{mksccf}%
\end{align}

\medskip

\noindent \textbf{Remark 2.} 
The order of the cofactor equation (\ref{mksccf}) is one less than the order of (\ref{dek}).
We can obtain the solution $t_{n+1}$ of the factor equation (\ref{kmscf}), which has order one,
either explicitly or determine its asymptotic properties. For this reason we may refer to 
Eq.(\ref{mksccf}) as an \textit{order reduction} for Eq.(\ref{dek}).

\subsubsection{Form symmetries and SC factorizations of separable equations}

Consider the \textit{additively separable} difference equation%
\begin{equation}
x_{n+1}=\alpha_{n}+\phi_{0}(x_{n})+\phi_{1}(x_{n-1})+\cdots\phi_{k}(x_{n-k})
\label{mde}%
\end{equation}
where $k\geq1$ is a fixed integer and (\ref{mde}) is defined over a nontrivial
\textit{subfield} $\mathcal{F}$ of $\mathbb{C}$, the set of complex numbers.
Specifically, we assume that
\begin{equation}
x_{-j},\alpha_{n}\in\mathcal{F},\quad\phi_{j}:\mathcal{F}\rightarrow
\mathcal{F},\quad j=0,1,\ldots,k,\ n=0,1,2,\ldots\label{mdec}%
\end{equation}

We use subfields of $\mathbb{C}$ rather than its subgroups because the results
below require closure under multiplication. The subfield that is most often
encountered in this context is the field of all real numbers $\mathbb{R}$. The
unfolding of Eq.(\ref{mde}) is defined as%
\[
F_{n}(z_{0},z_{1},\ldots,z_{k})=[\alpha_{n}+\phi_{0}(z_{0})+\cdots+\phi
_{k}(z_{k}),z_{0},\ldots,z_{k-1}].
\]

Based on the additive nature of Eq.(\ref{mde}) we look for a form symmetry of
type%
\begin{equation}
H(z_{0},z_{1},\ldots,z_{k})=z_{0}+h_{1}(z_{1})+\cdots+h_{k}(z_{k})
\label{fsad}%
\end{equation}
where $h_{j}:\mathcal{F}\rightarrow\mathbb{C}$ for all $j;$ we do not require
that $\mathcal{F}$ be invariant under the maps $h_{j}.$

With these assumptions, equality (\ref{scm}) takes the form%
\begin{equation}
\alpha_{n}+\phi_{0}(z_{0})+\cdots+\phi_{k}(z_{k})+h_{1}(z_{0})+\cdots
+h_{k}(z_{k-1})=g_{n}(z_{0}+h_{1}(z_{1})+\cdots+h_{k}(z_{k})) \label{sc1}%
\end{equation}
where $g_{n}:\mathbb{C}\rightarrow\mathbb{C}.$ Note that if (\ref{sc1}) holds
and the subfield $\mathcal{F}$ is invariant under $h_{j}$ for all $j$ then
$\mathcal{F}$ is invariant under $g_{n}$ also; i.e., $g_{n}(\mathcal{F}%
)\subset\mathcal{F}$ for all $n.$

Our aim is to determine the functions $h_{j}$ and $g_{n}$ that satisfy the
functional equation (\ref{sc1}). To simplify calculations first we assume that
the functions $\phi_{i}$ and $h_{j}$ are all differentiable on $\mathcal{F}$.
The differentiability assumption can be dropped once the form symmetry is calculated.

Take partial derivatives of both sides of (\ref{sc1}) for each of the $k+1$
variables $z_{0},z_{1},\ldots,z_{k}$ to obtain the following system of $k+1$
partial differential equations:
\begin{align}
\phi_{0}^{\prime}(z_{0})+h_{1}^{\prime}(z_{0})  &  =g_{n}^{\prime}(z_{0}%
+h_{1}(z_{1})+\cdots+h_{k}(z_{k}))\label{pde0}\\
\phi_{1}^{\prime}(z_{1})+h_{2}^{\prime}(z_{1})  &  =g_{n}^{\prime}(z_{0}%
+h_{1}(z_{1})+\cdots+h_{k}(z_{k}))h_{1}^{\prime}(z_{1})\label{pde1}\\
&  \vdots\nonumber\\
\phi_{k-1}^{\prime}(z_{k-1})+h_{k}^{\prime}(z_{k-1})  &  =g_{n}^{\prime}%
(z_{0}+h_{1}(z_{1})+\cdots+h_{k}(z_{k}))h_{k-1}^{\prime}(z_{k-1})\nonumber\\
\phi_{k}^{\prime}(z_{k})  &  =g_{n}^{\prime}(z_{0}+h_{1}(z_{1})+\cdots
+h_{k}(z_{k}))h_{k}^{\prime}(z_{k}). \label{pdek}%
\end{align}

Taking equations (\ref{pde0}) and (\ref{pde1}) together, we may eliminate
$g_{n}^{\prime}$ to obtain
\begin{equation}
\phi_{0}^{\prime}(z_{0})+h_{1}^{\prime}(z_{0})=\frac{\phi_{1}^{\prime}%
(z_{1})+h_{2}^{\prime}(z_{1})}{h_{1}^{\prime}(z_{1})}. \label{spde1}%
\end{equation}

The variables $z_{0}$ and $z_{1}$ are separated on different sides of PDE
(\ref{spde1}) so each side must be equal to a constant $c\in\mathbb{C}.$ The
left hand side of (\ref{spde1}) gives%
\begin{equation}
\phi_{0}^{\prime}(z)+h_{1}^{\prime}(z)=c\Rightarrow h_{1}(z)=cz-\phi_{0}(z).
\label{ah1}%
\end{equation}

Similarly, from the right hand side of (\ref{spde1}) we get%

\[
\frac{\phi_{1}^{\prime}(z)+h_{2}^{\prime}(z)}{h_{1}^{\prime}(z)}=c\Rightarrow
h_{2}(z)=ch_{1}(z)-\phi_{1}(z).
\]

Note that from (\ref{pde0}) and (\ref{ah1}) it follows that%
\begin{equation}
g_{n}^{\prime}(z_{0}+h_{1}(z_{1})+\cdots+h_{k}(z_{k}))=c. \label{gn}%
\end{equation}

It is evident that this calculation can be repeated for equations 3 to $k-1$
in the above system of PDE's and the following functional recursion is
established by induction%
\begin{equation}
h_{j}=ch_{j-1}-\phi_{j-1},\quad j=1,\ldots,k,\ h_{0}(z)\doteq z. \label{hj}%
\end{equation}

Finally, the last PDE (\ref{pdek}) gives%
\begin{equation}
ch_{k}=\phi_{k}. \label{hk}%
\end{equation}

On the other hand, setting $j=k$ in (\ref{hj}) gives $h_{k}=ch_{k-1}%
-\phi_{k-1}$ so equality (\ref{hk}) implies that there must be a restriction
on the functions $\phi_{j}.$ To determine this restriction precisely, first
notice that by (\ref{hk}) and (\ref{hj})
\[
c^{2}h_{k-1}(z)-c\phi_{k-1}(z)-\phi_{k}(z)=0\quad\text{for all }z.
\]

Now, applying (\ref{hj}) repeatedly $k-1$ more times removes the functions
$h_{j}$ to give the following identity%
\begin{equation}
c^{k+1}z-c^{k}\phi_{0}(z)-c^{k-1}\phi_{1}(z)-\cdots-c\phi_{k-1}(z)-\phi
_{k}(z)=0\quad\text{for all }z. \label{lc}%
\end{equation}

This equality shows that the existence of form symmetries of type (\ref{fsad})
requires that the given maps $\phi_{0},\phi_{1},\ldots,\phi_{k}$ plus the
identity function form a linearly dependent set. Equivalently, (\ref{lc})
determines any one of the functions $\phi_{0},\phi_{1},\ldots,\phi_{k}$ as a
linear combination of the identity function and the other $k$ functions; e.g.,%
\begin{equation}
\phi_{k}(z)=c^{k+1}z-c^{k}\phi_{0}(z)-c^{k-1}\phi_{1}(z)-\cdots-c\phi
_{k-1}(z). \label{lck}%
\end{equation}

Thus if there is $c\in\mathbb{C}$ such that (\ref{lc}) or (\ref{lck}) hold
then Eq.(\ref{mde}) has a form symmetry of type (\ref{fsad}). This form
symmetry $H$ can now be calculated as follows: Using the recursion (\ref{hj})
repeatedly gives
\begin{equation}
h_{j}(z)=c^{j}z-c^{j-1}\phi_{0}(z)-\cdots-\phi_{j-1}(z),\quad j=1,\ldots k.
\label{hjz}%
\end{equation}

Therefore,%
\begin{align*}
H(z_{0},z_{1},\ldots,z_{k})  &  =z_{0}+h_{1}(z_{1})+\cdots+h_{k}(z_{k})\\
&  =z_{0}+\sum_{j=1}^{k}[c^{j}z_{j}-c^{j-1}\phi_{0}(z_{j})-\cdots-\phi
_{j-1}(z_{j})].
\end{align*}

With the form symmetry calculated, we now proceed with the SC factorization of
Eq.(\ref{mde}). To determine the first order factor equation, from (\ref{gn})
we gather that $g_{n}(z)$ is a linear non-homogeneous function of $z$. The
precise formula for $g_{n}$ can be determined from (\ref{sc1}), (\ref{hj}) and
(\ref{hk}) as follows: First, note that if the constant $c$ is in the subfield
$\mathcal{F}$ then by (\ref{hjz}) $\mathcal{F}$ is invariant under all $h_{j}$
and thus also under all $g_{n}.$ In this case, the function $H(z_{0}%
,z_{1},\ldots,z_{k})$ is onto $\mathcal{F},$ i.e., $H(\mathcal{F}%
^{k+1})=\mathcal{F}.$ This is easy to see, since for each $z\in\mathcal{F}$
and arbitrary $z_{1},\ldots,z_{k}\in\mathcal{F}$ we may choose $z_{0}%
=z-[h_{1}(z_{1})+\cdots+h_{k}(z_{k})]\in\mathcal{F}$ so that $z=H(z_{0}%
,z_{1},\ldots,z_{k}).$ On the other hand, if $c\notin\mathcal{F}$ then a
similar argument proves that $H$ is onto $\mathbb{C}.$ Now, we calculate
$g_{n}$ on $\mathcal{F}$ or on all of $\mathbb{C}$ as
\begin{align*}
g_{n}(z)  &  =g_{n}(z_{0}+h_{1}(z_{1})+\cdots+h_{k}(z_{k}))\\
&  =\alpha_{n}+\sum_{j=1}^{k}[\phi_{j-1}(z_{j-1})+h_{j}(z_{j-1})]+\phi
_{k}(z_{k})\\
&  =\alpha_{n}+c[z_{0}+h_{1}(z_{1})+\cdots+h_{k}(z_{k})]\\
&  =\alpha_{n}+cz.
\end{align*}

This determines the factor equation (\ref{mkscf}) for Eq.(\ref{mde}); the
cofactor equation is given by (\ref{mksccf}). Therefore, the SC factorization
is%
\begin{align}
z_{n+1}  &  =\alpha_{n}+cz_{n},\quad z_{0}=x_{0}+h_{1}(x_{-1})+\cdots
+h_{k}(x_{-k})\label{lnh}\\
x_{n+1}  &  =z_{n+1}-h_{1}(x_{n})-\cdots-h_{k}(x_{n-k+1}). \label{ro}%
\end{align}

Note that the functions $h_{j}$ do not involve the derivatives of $\phi_{i}$
and we showed that the form symmetry $H$ is a semiconjugate link without using
derivatives. Therefore, the SC factorization is valid even if $\phi_{i}$ are
not differentiable. The following result summarizes the discussion above.

\begin{theorem}
\label{septhm} \noindent Eq.(\ref{mde}) has a separable form symmetry
(\ref{fsad}) if there is $c\in\mathbb{C}$\ such that (\ref{lc}) holds. In this
case, (\ref{mde}) reduces to the system of equations (\ref{lnh}) and
(\ref{ro}) of lower orders with functions $h_{j}$\ given by (\ref{hjz}).
\end{theorem}

\subsubsection{Complete factorization of linear equations}

A significant feature of Eq.(\ref{ro}), which by Remark 2 represents a reduction
of order for Eq.(\ref{mde}), is that (\ref{ro}) is of the same separable
type as (\ref{mde}). Thus if the functions $h_{1},\ldots,h_{k}$ satisfy the
analog of (\ref{lc}) for some constant $c^{\prime}\in\mathbb{C}$ then Theorem
\ref{septhm} can be applied to (\ref{ro}). In the important case of linear
difference equations this process can be continued until we are left with a
system of first order linear equations, as seen in the next result (for its
proof and related comments see \cite{hsarx})

\begin{corollary}
\noindent\label{lincor}\textit{The linear non-homogeneous difference equation
of order }$k+1$\textit{\ with constant coefficients}%
\begin{equation}
x_{n+1}+b_{0}x_{n}+b_{1}x_{n-1}+\cdots+b_{k}x_{n-1}=\alpha_{n} \label{LDE}%
\end{equation}
\textit{where }$b_{0},\ldots,b_{k}$\textit{, }$a_{n}\in\mathbb{C}$\textit{\ is
equivalent to the following system of }$k+1$\textit{\ first order linear
non-homogeneous equations}%
\begin{align*}
z_{0,n+1}  &  =\alpha_{n}+c_{0}z_{0,n},\\
z_{1,n+1}  &  =z_{0,n+1}+c_{1}z_{1,n}\\
&  \vdots\\
z_{k,n+1}  &  =z_{k-1,n+1}+c_{k}z_{k,n}%
\end{align*}
\textit{in which }$z_{k,n}=x_{n}$\textit{\ is the solution of Eq.(\ref{LDE})
and the complex constants }$c_{0},c_{1},\ldots,c_{k}$\textit{\ are the
eigenvalues of the homogeneous part of (\ref{LDE}), i.e., they are roots of
the characteristic polynomial}%
\begin{equation}
P(z)=z^{k+1}+b_{0}z^{k}+b_{1}z^{k-1}+\cdots+b_{k-1}z+b_{k}. \label{cp}%
\end{equation}

\end{corollary}

The preceding result shows that the classical reduction of order technique
through operator factorization is also a consequence of semiconjugacy. Indeed,
even the eigenvalues can be explained via semiconjugacy. Corollary
\ref{lincor} also shows that the complete factorization of the linear equation
(\ref{LDE}) results in a cofactor chain of length $k+1.$

\subsubsection{Multiplicative forms}

As another application of Theorem \ref{septhm} we consider the following
difference equation on the positive real line%
\begin{gather}
y_{n+1}=\beta_{n}\psi_{0}(y_{n})\psi_{1}(y_{n-1})\cdots\psi_{k}(y_{n-k}%
),\label{mult}\\
\beta_{n},y_{-j}\in(0,\infty),\ \psi_{j}:(0,\infty)\rightarrow(0,\infty
),\ j=0,\ldots k.\nonumber
\end{gather}

This difference equation is separable in a multiplicative sense, i.e., with
$(0,\infty)$ as a multiplicative group. Taking the logarthim of Eq.(\ref{mult}%
) changes it to an additive equation. Specifically, by defining%
\[
x_{n}=\ln y_{n},\ y_{n}=e^{x_{n}},\ \ln\beta_{n}=\alpha_{n},\ \phi_{j}%
(r)=\ln\psi_{j}(e^{r}),\ j=0,\ldots k,\ r\in\mathbb{R}%
\]
we can transform (\ref{mult}) into (\ref{mde}). Then Theorem \ref{septhm}
implies the following:

\begin{corollary}
\label{sepmult} \noindent\textit{Eq.(\ref{mult}) has a form symmetry }%
\begin{equation}
H(t_{0},t_{1},\ldots,t_{k})=t_{0}h_{1}(t_{1})\cdots h_{k}(t_{k}) \label{Hmult}%
\end{equation}
\textit{if there is }$c\in\mathbb{C}$\textit{\ such that the following is true
for all }$t>0,$%
\begin{equation}
\psi_{0}(t)^{c^{k}}\psi_{1}(t)^{c^{k-1}}\cdots\psi_{k}(t)=t^{c^{k+1}}.
\label{lcmult}%
\end{equation}
\textit{The functions }$h_{j}$ \textit{in (\ref{Hmult}) are given as}%
\begin{equation}
h_{j}(t)=t^{c^{j}}\psi_{0}(t)^{-c^{j-1}}\cdots\psi_{j-1}(t),\quad j=1,\ldots k
\label{hjmult}%
\end{equation}
\textit{and the form symmetry in (\ref{Hmult}) and (\ref{hjmult}) yields the
type-}$(1,k)$\textit{\ order reduction}
\begin{align}
r_{n+1}  &  =\beta_{n}r_{n}^{c},\quad r_{0}=y_{0}h_{1}(y_{-1})\cdots
h_{k}(y_{-k})\label{fmult}\\
y_{n+1}  &  =\frac{r_{n+1}}{h_{1}(y_{n})\cdots h_{k}(y_{n-k+1})}.
\label{flink}%
\end{align}

\end{corollary}

For a different and direct proof of the above corollary when $k=1$ see
\cite{hsijpam}. The next example applies Corollary \ref{sepmult} to an
interesting difference equation. In particular, it shows that the global
behavior of a population model discussed in \cite{fhl} is much more varied
than previously thought.

\begin{example}
\label{multexp}\textbf{(A simple equation with complicated multistable
solutions)}. Equations of type (\ref{mde}) or (\ref{mult}) are capable of
exhibiting complex behavior, e.g., coexisting stable solutions of
many different types that range from periodic to chaotic. Consider the
following second-order equation
\begin{equation}
x_{n+1}=x_{n-1}e^{a-x_{n}-x_{n-1}},\quad x_{-1},x_{0}>0.\label{exp}%
\end{equation}
Note that Eq.(\ref{exp}) has up to two isolated fixed points. One is the
origin which is repelling if $a>0$ (eigenvalues of linearization are $\pm
e^{a/2}$) and the other fixed point is $\bar{x}=a/2$. If $a>4$ then $\bar{x}$
is unstable and non-hyperbolic because the eigenvalues of the linearization of
(\ref{exp}) are $-1$ and $1-a/2$. Figure 1 shows the bifurcation of
numerically generated solutions of this equation with $a=4.6$. One initial
value $x_{-1}=2.3$ is fixed and the other initial value $x_{0}$ ranges from
2.3 to 4.8; i.e., approaching (or moving away from) the fixed point $\bar
{x}=2.3$ on a straight line segment in the plane.

\begin{figure}[ptb]
\begin{center}
\includegraphics[width=4.795in]{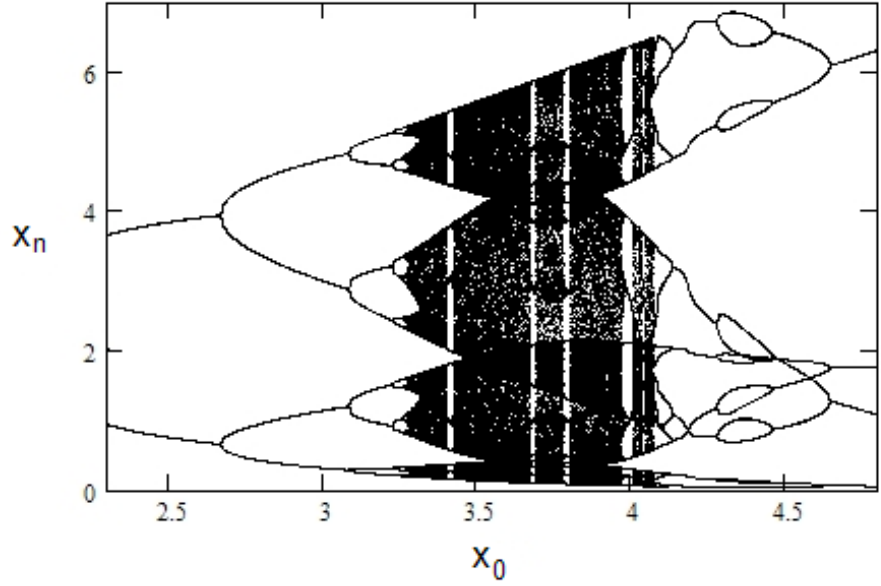}
\end{center}
\caption{{\small {Bifurcations of solutions of Eq.(\ref{exp}) with a changing
initial value; $a=4.6$ is fixed.}}}%
\end{figure}

Stable solutions with periods 2, 4, 8, 12 and 16 can be easily identified in
Figure 1. All of the solutions that appear in this figure represent
\textit{coexisting stable orbits} of Eq.(\ref{exp}). There are also periodic
and non-periodic solutions which do not appear in Figure 1 because they are
unstable (e.g., the fixed point $\bar{x}=2.3$). Additional bifurcations of
both periodic and non-periodic solutions occur outside the range 2.3-4.8 which
are not shown.

Understanding the global behaviors of solutions of Eq.(\ref{exp}) is
relatively easy if we examine its SC factorization given by (\ref{fmult}) and
(\ref{flink}). Here $k=1$ and
\[
\psi_{0}(t)=e^{-t},\quad\psi_{1}(t)=te^{-t},\quad\beta_{n}=e^{a}\text{ for all
}n.
\]
Thus (\ref{lcmult}) takes the form%
\begin{align*}
\psi_{0}(t)^{c}\psi_{1}(t) &  =t^{c^{2}}\text{ for all }t>0\\
e^{-ct}te^{-t} &  =t^{c^{2}}\text{ for all }t>0
\end{align*}
The last equality is true if $c=-1,$ which leads to the form symmetry%
\[
h_{1}(t)=t^{-1}\psi_{0}(t)^{-1}=\frac{1}{te^{-t}}\Rightarrow H(u_{0}%
,u_{1})=\frac{u_{0}}{u_{1}e^{-u_{1}}}%
\]
and SC factorization%
\begin{align}
r_{n+1} &  =\frac{e^{a}}{r_{n}},\quad r_{0}=x_{0}h_{1}(x_{-1})=\frac{x_{0}%
}{x_{-1}e^{-x_{-1}}}\label{star1}\\
x_{n+1} &  =\frac{r_{n+1}}{h_{1}(x_{n})}=r_{n+1}x_{n}e^{-x_{n}}.\label{star2}%
\end{align}
All positive solutions of (\ref{star1}) with $r_{0}\neq e^{a/2}$ are periodic
with period 2:%
\[
\left\{  r_{0},\frac{e^{a}}{r_{0}}\right\}  =\left\{  \frac{x_{0}}%
{x_{-1}e^{-x_{-1}}},\frac{x_{-1}e^{a-x_{-1}}}{x_{0}}\right\}  .
\]
Hence the orbit of each nontrivial solution $\{x_{n}\}$ of (\ref{exp}) in the
plane is restricted to the pair of curves%
\begin{equation}
\xi_{1}(t)=\frac{e^{a}}{r_{0}}te^{-t}\quad\text{and\quad}\xi_{2}%
(t)=r_{0}te^{-t}.\label{ic}%
\end{equation}
Now, if $x_{-1}$ is fixed and $x_{0}$ changes, then $r_{0}$ changes
proportionately to $x_{0}$. These changes in initial values are reflected as
changes in \textit{parameters} in (\ref{star2}). The orbits of the one
dimensional map $bte^{-t}$ where $b=r_{0}$ or $e^{a}/r_{0}$\ exhibit a variety
of behaviors as the parameter $b$ changes according to well-known rules such
as the fundamental ordering of cycles and the occurrence of chaotic behavior
with the appearance of period-3 orbits when $b$ is large enough.
Eq.(\ref{star2}) splits the orbits evenly over the pair of curves (\ref{ic})
as the initial value $x_{0}$ changes and coexisting, qualitatively different
stable solutions appear. \ 

\begin{figure}[ptb]
\begin{center}
\includegraphics[width=4.557in]{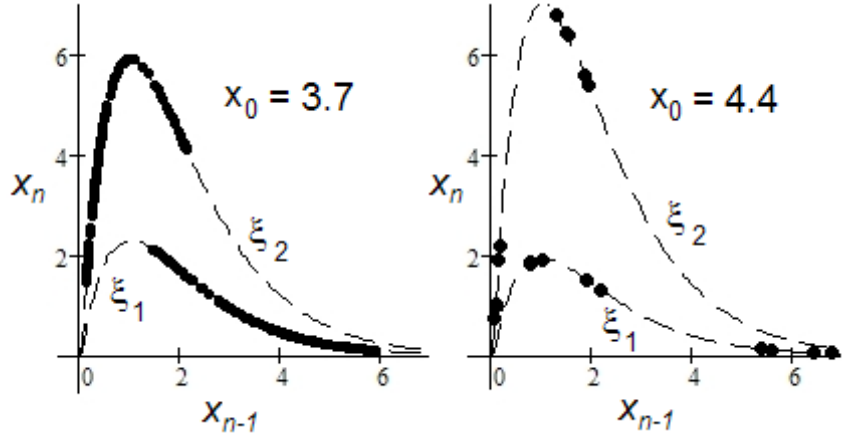}
\end{center}
\caption{{\small {Two of the orbits in Figure 1 shown here on their loci of
two curves $\xi_{1}$, $\xi_{2}$ in the state space.} }}%
\end{figure}

Figure 2 shows the orbits of (\ref{exp}) for two different initial values
$x_{0}$ with $a=4.6;$ the first 100 points of each orbit are discarded in
these images so as to highlight the asymptotic behavior of each orbit. 
\end{example}

\section{Concluding remarks}

In this article we covered the basic theory of reduction of order in higher
order difference equations by semiconjugate factorization. Needless to say, a
great deal more remains to be done in this area before a sufficient level of
maturity is reached. It is necessary to generalize and extend various ideas
discussed above. In addition, a more detailed theory is necessary for a full discussion of type-$(m,k)$ reductions for all $m$. 

The list of things needing proper attention is clearly too long for a research article. An upcoming book \cite{fsorbk} is devoted to the topic of semiconjugate factorizations and reduction of order. This book consolidates the existing literature on the subject and offers in-depth discussions, more
examples, additional theory and lists of problems for further study.

\end{document}